\documentclass{elsarticle}

\usepackage[T1]{fontenc}
\usepackage[latin1]{inputenc}
\usepackage[english]{babel}
\usepackage{amsthm}
\usepackage{amsmath}
\usepackage{amssymb}
\usepackage{amscd}
\usepackage{vmargin}
\usepackage[colorlinks=true,citecolor=black,linkcolor=black,urlcolor=blue]{hyperref}
\usepackage{cleveref}
\usepackage{subfig}
\usepackage{nicefrac}
\usepackage{xspace}
\setmarginsrb{3cm}{2cm}{3cm}{2cm}{1cm}{1cm}{1cm}{1cm} 
\newtheorem{theorem}{Theorem}
\crefname{theorem}{Theorem}{Theorems}
\newtheorem{corollary}[theorem]{Corollary}
\crefname{corollary}{Corollary}{Corollaries}

\newtheorem{proposition}[theorem]{Proposition}
\crefname{proposition}{Proposition}{Propositions}

\crefname{claim}{claim}{claims}
\Crefname{claim}{Claim}{Claims}
\newtheorem{remark}[theorem]{Remark}
\crefname{remark}{Remark}{Remarks}

\newcommand{\vt}[1]{\overrightarrow{#1}}

\newcommand{\ie}{\textit{i.e.}\xspace}

\def\paren#1{\left( #1 \right)}
\def\acc#1{\left\{ #1 \right\}}

\DeclareMathOperator{\at}{at}

\renewcommand{\le}{\leqslant}
\renewcommand{\ge}{\geqslant}
\newcommand{\pnk}{$P_{n,k}$\xspace}

\begin{document}
\begin{frontmatter}
  \title{Oriented coloring of graphs with low maximum degree\tnoteref{t1}}
  \tnotetext[t1]{The work was partially supported by the French ANR grant HOSIGRA
    under contract ANR-17-CE40-0022.} 
  \author{Pascal Ochem} 
  \author{Alexandre Pinlou}
  \address{LIRMM, Université de Montpellier, CNRS, Montpellier,
    France}
  
  \begin{abstract} 
    Duffy et al. [C. Duffy, G. MacGillivray, and É. Sopena, \textit{Oriented
    colourings of graphs with maximum degree three and four},
    Discrete Mathematics, 342(4), p. 959--974, 2019] recently
    considered the oriented chromatic number of \emph{connected}
    oriented graphs with maximum degree~$3$ and~$4$, proving it is at
    most $9$ and $69$, respectively. In this paper, we improve these
    results by showing that the oriented chromatic number of
    \emph{non-necessarily connected} oriented graphs with maximum
    degree $3$ (resp.~$4$) is at most $9$ (resp. $26$). The bound of
    $26$ actually follows from a general result which determines properties of a target
    graph to be universal for graphs of bounded maximum degree.
    This generalization also allows us to get 
    the upper bound of $90$ (resp. $306$, $1322$) for
    the oriented chromatic number of graphs with maximum degree $5$
    (resp. $6$, $7$).
  \end{abstract}

  \begin{keyword}
    Oriented graph, Homomorphism, Coloring, Bounded degree graph.
  \end{keyword}

\end{frontmatter}

\section{Introduction}

Oriented graphs are directed graphs with neither loops nor opposite arcs.
Unless otherwise specified, the term \emph{graph} refers to \emph{oriented graph} in the sequel.

For a graph~$G$, we denote by $V(G)$ its set of vertices and
by $A(G)$ its set of arcs. For two adjacent vertices $u$ and $v$, we
denote by $\vt{uv}$ the arc from $u$ to $v$, or simply $uv$ whenever
its orientation is not relevant (therefore, $uv=\vt{uv}$ or
$uv=\vt{vu}$).

Given two graphs $G$ and $H$, a \emph{homomorphism} from $G$
to $H$ is a mapping $\varphi:V(G)\rightarrow V(H)$ that preserves the
arcs, that is, $\vt{\varphi(x)\varphi(y)}\in A(H)$ whenever $\vt{xy}\in A(G)$.

An \emph{oriented $k$-coloring} of $G$ can be defined as a homomorphism from
$G$ to $H$, where $H$ is a graph with $k$ vertices. The existence of such a
homomorphism from $G$ to $H$ is denoted by $G\rightarrow H$. The
vertices of $H$ are called \emph{colors}, and we say that $G$ is
$H$-colorable. The \emph{oriented chromatic number} of a
graph $G$, denoted by $\chi_o(G)$, is defined as the smallest number
of vertices of a graph $H$ such that $G\rightarrow H$.

The notion of oriented coloring introduced by
Courcelle~\cite{courcelle_monadic_1994} has been studied by several
authors and the problem of bounding the oriented chromatic number has
been investigated for various graph classes: outerplanar graphs (with
given minimum girth)~\cite{pinlou_oriented_2006-3,
  sopena_chromatic_1997}, 2-outerplanar
graphs~\cite{esperet_oriented_2007, ochem_oriented_2014}, planar
graphs (with given minimum girth)~\cite{borodin_oriented_2005,
  borodin_oriented_2005-1, borodin_oriented_2007,
  borodin_maximum_1999, ochem_oriented_2004, ochem_oriented_2014,
  pinlou_oriented_2009, raspaud_good_1994}, graphs with bounded
maximum average degree~\cite{borodin_oriented_2007,
  borodin_maximum_1999}, graphs with bounded
degree~\cite{duffy_oriented_2019, kostochka_acyclic_1997,
  sopena_note_1996}, graphs with bounded
treewidth~\cite{ochem_oriented_2008, sopena_chromatic_1997,
  sopena_oriented_2001}, graph
subdivisions~\cite{wood_acyclic_2005}, \dots A survey on the
study of oriented colorings has been done by Sopena in 2001 and
recently updated~\cite{sopena_homomorphisms_2016}. 

For bounded degree graphs, Kostochka et
al.~\cite{kostochka_acyclic_1997} proved as a general bound that graphs with maximum
degree~$\Delta$ have oriented chromatic number at most
$2\Delta^22^\Delta$. They also showed that, for every $\Delta$, there exists 
graphs with maximum degree $\Delta$ and oriented chromatic number at
least~$2^{\nicefrac{\Delta}{2}}$. For low maximum degrees, specific
results
are only known for graphs with maximum degree~$3$ and~$4$.
Sopena~\cite{sopena_chromatic_1997} proved that graphs with
maximum degree~$3$ have an oriented chromatic number at most~$16$ and
conjectured that any such connected graphs have an  oriented 
chromatic number at most~$7$. The upper bound was later improved by Sopena and
Vignal~\cite{sopena_note_1996} to~$11$. Recently, Duffy et
al.~\cite{duffy_oriented_2019} proved that~$9$ colors are enough for
\emph{connected} graphs with maximum degree~$3$. They proved
in the same paper that \emph{connected} graphs with maximum degree~$4$ have
oriented chromatic number at most~$69$. Lower bounds
are given by Duffy et al.~\cite{duffy_oriented_2019} who 
exhibit a graph with maximum degree~$3$ (resp.~$4$) and oriented
chromatic number~$7$ (resp.~$11$). Note that the above-mentioned
conjecture of Sopena is thus best possible.

In this paper, we improve the upper bounds of graph with maximum
degree~$3$ and~$4$. We first prove in \Cref{sec:deg3} that the
oriented chromatic number of graphs with maximum degree~$3$ is at
most~$9$, that is, we remove the condition of connectivity;
see~\Cref{thm:deg3}. For graphs with maximum degree $4$, we in fact
propose in~\Cref{sec:deg4} a general result which determines
properties of a target graph to be universal for (non-necessarily
connected) graphs of maximum degree $\Delta\ge 4$;
see~\Cref{thm:deg5-8}. As a consequence of this general result, we get
that the oriented chromatic number of graph with maximum degree~$4$ is
at most~$26$, substantially decreasing the bound of $69$ due to Duffy
et al.~\cite{duffy_oriented_2019}. We also get that the oriented
chromatic number of graphs with maximum degree $5$ (resp. $6$, $7$) is
at most $90$ (resp. $306$, $1322$). The next two sections will be devoted to define
the notation and to present the properties of the target graphs we use
to prove our results.

\section{Notation}

In the remainder of this paper, we use the following notions. For a
vertex $v$ of a graph $G$, we denote by $N^+_G(v)$ the set of outgoing
neighbors of $v$, by $N^-_G(v)$ the set of incoming neighbors of $v$
and by $N_G(v) = N^+_G(v) \cup N^-_G(v)$ the set of neighbors of $v$
(subscripts are omitted when the considered graph is clearly
identified from the context). The \emph{degree} of a vertex $v$ is the
number of its neighbors $|N(v)|$. If two graphs $G$ and $H$ are isomorphic, we denote
this by~$G\cong H$.

\section{Paley tournaments and Tromp digraphs} 

In this section, we describe the general construction of graphs
that will be used to prove~\Cref{thm:deg3,thm:deg5-8}, and present some
of their useful properties.

For a prime power $p\equiv 3\pmod 4$, the \emph{Paley tournament}
$QR_p$ is defined as the graph whose vertices are the elements of the
field $\mathbb{F}_p$ and such that $\vt{uv}$ is an arc if and only if
$v - u$ is a non-zero quadratic residue of $\mathbb{F}_p$.  Clearly
$QR_p$ is vertex- and arc-transitive. 

An \emph{orientation $n$-vector} is a sequence
$\alpha = (\alpha_1,\alpha_2,\ldots,\alpha_n)\in \{-1,1\}^n$ of $n$
elements. Let $S=(v_1,v_2,\ldots,v_n)$ be a sequence of $n$ (not
necessarily distinct) vertices of a graph $G$. The vertex $u$ is said
to be an \emph{$\alpha$-successor of~$S$} if for any $i$,
$1\le i\le n$, we have $\vt{uv_i}\in A(G)$ whenever $\alpha_i = 1$ and
$\vt{v_iu}\in A(G)$ otherwise. A sequence $S=(v_1,v_2,\ldots,v_n)$ of
$n$ (not necessarily distinct) vertices of $QR_p$ is said to be
\emph{compatible} with an orientation $n$-vector
$\alpha=(\alpha_1,\alpha_2,$ $\ldots,\alpha_n)$ if and only if
$\alpha_i=\alpha_j$ whenever $v_i = v_j$ since graphs do not contain
opposite arcs. We say that a graph $G$ has \emph{Property \pnk} 
if, for every sequence $S$ of $n$ vertices of $G$ and any compatible
orientation $n$-vector~$\alpha$, there exist $k$ distinct $\alpha$-successors
of $S$. Such Properties \pnk have been extensively used in many papers
dealing with oriented coloring. 

\begin{proposition}\label{prop:qr11}
  The Paley tournament $QR_p$ has Properties $P_{1,\frac{p-1}2}$ and $P_{2,\frac{p-3}4}$.
\end{proposition}

\begin{proof}
  By the vertex-transitivity of $QR_p$, the in-degree of every vertex
  is equal to its out-degree.  This implies that $QR_p$ has Property
  $P_{1,\frac{p-1}2}$.

  Let us prove that $QR_p$ has Property $P_{2,\frac{p-3}4}$. To do so,
  by arc-transitivity of $QR_p$, we just have to show that there exist
  at least $\frac{p-3}4$ $\alpha$-successors of the sequence $S=(0,1)$
  for any of the four orientation vector $\alpha \in \{-1,1\}^2$.

  We first need to count the transitive triangles with arcs $\vt{xy}$,
  $\vt{yz}$, and $\vt{xz}$ in $QR_p$.  There are $p$ choices for the
  source vertex $x$ of a transitive triangle.  The number of
  transitive triangles such that $x=0$ is equal to the number of arcs
  in $N^+(0)$, that is,
  $\binom{\nicefrac{(p-1)}2}{2}=\frac{(p-1)(p-3)}8$.  Thus, $QR_p$
  contains $\frac{p(p-1)(p-3)}8$ transitive triangles.

  Considering $\alpha=(+1,+1)$, we can notice that
  $|N^+(0)\cap N^+(1)|$ is the number of transitive triangles such
  that $\vt{xy}=\vt{01}$.  Since $QR_p$ is arc-transitive and contains
  $\frac{p(p-1)}2$ arcs,
  $|N^+(0)\cap N^+(1)|=\frac{p(p-1)(p-3)/8}{p(p-1)/2} = \frac{p-3}4$.
  Similarly for $\alpha=(-1,-1)$ and $\alpha=(+1,-1)$, considering
  $\vt{yz}=\vt{01}$ gives $|N^-(0)\cap N^-(1)| = \frac{p-3}4$ and
  considering $\vt{xz}=\vt{01}$ gives
  $|N^+(0)\cap N^-(1)| = \frac{p-3}4$.  Finally for $\alpha=(-1,+1)$,
  we have
  $N^-(0)\cap N^+(1) = V\paren{QR_p}\setminus\{0,1,N^+(0)\cap
  N^+(1),N^-(0)\cap N^-(1), N^+(0)\cap N^-(1)\}$, so that
  $|N^-(0)\cap N^+(1)| = p-2-\frac{3(p-3)}4 = \frac{p+1}4 > \frac{p-3}4$.  This
  proves $P_{2,\frac{p-3}4}$.
\end{proof}

Paley tournaments will be used as basic brick to build new graphs as
explained below. Tromp (unpublished manuscript) proposed the following
construction.  Let $G$ be a graph and let $G'\cong G$.  The Tromp
graph $Tr(G)$ has $2|V(G)|+2$ vertices and is defined as follows:

\begin{itemize}
\item $V(Tr(G)) = V(G) \cup V(G') \cup \{\infty,\infty'\}$ 
\item $\forall u\in V(G) : \vt{u\infty}, \vt{\infty u'},
  \vt{u'\infty'}, \vt{\infty' u}\in A(Tr(G))$ 
\item $\forall u,v\in V(G), \vt{uv}\in A(G) : \vt{uv}, \vt{u'v'},
  \vt{vu'}, \vt{v'u} \in A(Tr(G))$
\end{itemize}

\begin{figure}
  \begin{center}
    \includegraphics[scale=2]{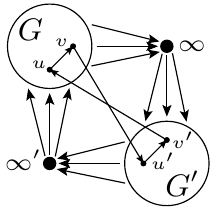}
  \end{center}
  \caption{The Tromp graph $Tr(G)$.\label{fig:tromp}}
\end{figure}

\Cref{fig:tromp} illustrates the construction of $Tr(G)$. We can
observe that, for every $u\in V(G) \cup \{\infty\}$, there is no arc
between $u$ and $u'$. Such pairs of vertices will be called
\emph{anti-twin vertices}, and we denote by $\at(u)=u'$ the anti-twin
vertex of $u$.

In the following, we apply Tromp's construction to Paley tournaments
$QR_p$ which produces graphs with interesting structural
properties. First of all, Marshall~\cite{marshall_homomorphism_2007}
proved that any $Tr(QR_p)$ is \emph{vertex-transitive} and
\emph{arc-transitive}. He also prove that any $Tr(QR_p)$ is
\emph{triangle-transitive}, meaning that, given two triangles
$u_1u_2u_3$ and $v_1v_2v_3$ of $Tr(QR_p)$ with the same orientation,
there exists an automorphism that maps $u_i$ to $v_i$.
Secondly, it is possible to derive Properties \pnk for
$Tr(QR_p)$ knowing those of $QR_p$ (see~\Cref{prop:pnk}). Let us first
define the notion of compatible sequence of vertices of $Tr(QR_p)$
with an orientation vector: A sequence $S=(v_1,v_2,\ldots,v_n)$ of $n$
(not necessarily distinct) vertices of $Tr(QR_p)$ is said to be
\emph{compatible} with an orientation $n$-vector
$\alpha = (\alpha_1,\alpha_2,\ldots,\alpha_n)$ if and only if for any
$i\ne j$, we have $\alpha_i\ne\alpha_j$ whenever $v_i = \at(v_j)$, and
$\alpha_i=\alpha_j$ whenever $v_i = v_j$. If the $n$ vertices of $S$
induce an $n$-clique subgraph of $Tr(QR_p)$ (\ie $v_1,v_2,\ldots,v_n$
are pairwise distinct and induce a complete graph), then $S$ is
compatible with any orientation $n$-vector since a vertex $u$ and its
anti-twin $\at(u)$ cannot belong together to the same clique. The
authors already studied properties of $Tr(QR_{19})$
(see~\cite[Proposition~5]{ochem_oriented_2014}) and their results can
be easily generalized to $Tr(QR_p)$:

\begin{proposition}\cite{ochem_oriented_2014}\label{prop:pnk} 
  If $QR_p$ has Property $P_{n-1,k}$, then $Tr(QR_p)$ has Property \pnk.
\end{proposition}

Let us now introduce another type of properties.  We will say that a
graph $G$ has \emph{Property $C_{n,k}$} if, given $n$ vertices
$v_1,v_2,\dots,v_n$ of $G$ that form a clique subgraph, we have
$|\bigcup_{1\le i\le n} N^+(v_i)| \ge k$ and
$|\bigcup_{1\le i\le n} N^-(v_i)| \ge k$.

\begin{remark}\label{rem:cnk}
  Given two integers $n$ and $k$, a graph having Property $C_{n,k}$
  has Property $C_{n',k'}$ for any $n'$ and $k'$ such that $n'\ge n$
  and $k'\le k$.
\end{remark}

\begin{proposition}\label{lem:17}
  The graph $Tr(QR_p)$ has Properties $C_{2,\frac{3p+1}2}$ and
  $C_{3,\frac{7p+3}4}$.
\end{proposition}

\begin{proof}
  Recall that $Tr(QR_p)$ is built from two copies of $QR_p$, that will
  be named $QR_p$ and $QR'_p$ in the following 
  (see~\Cref{fig:tromp}). In this proof, the
  in- and out-neighborhood of a vertex $v$ of $Tr(QR_p)$ will be
  denoted by $N^-(v)$ and $N^+(v)$, while the
  in- and out-neighborhood of a vertex $v$ in a subgraph $H$ of
  $Tr(QR_p)$ will be denoted by $N_H^-(v)$ and $N_H^+(v)$.
  
  Note that, given a set of $n$ vertices $v_1,v_2,\dots,v_n$ of
  $Tr(QR_p)$ that form a clique subgraph, if
  $z\in \bigcup_{1\le i\le n}N^+(v_i)$ then
  $\at(z) = z'\in \bigcup_{1\le i\le n} N^-(v_i)$. Thus
  $|\bigcup_{1\le i\le n} N^+(v_i)| = |\bigcup_{1\le i\le n}
  N^-(v_i)|$.

  \begin{itemize}
  \item Let us first consider Property $C_{2,\frac{3p+1}2}$. We have
    to prove that given two adjacent vertices $x$ and $y$ of
    $Tr(QR_p)$, we have
    $|N^+(x)\cup N^+(y)|\ge \frac{3p+1}2$.
    
    Since $x$ and $y$ are adjacent, w.l.o.g. $x=0$ and $y=\infty$ by
    arc-transitivity of $Tr(QR_p)$. Then $N^+(0)\cup N^+(\infty)$
    contains:
    \begin{itemize}
    \item $N^+(\infty)=\acc{0',1',\dots,(p-1)'}$ ($p$ vertices);
    \item Note that
      $N^+(0) = N_{QR_p}^+(0) \uplus N_{QR'_p}^+(0) \uplus
      \acc{\infty}$. Since $N_{QR'_p}^+(0) \subset N^+(\infty)$ is
      already counted in the previous point, we just consider
      $N_{QR_p}^+(0)$ (at least $\frac{p-1}{2}$ vertices by
      \Cref{prop:qr11}) and $\infty$ (1 vertex);
    \end{itemize}
    So that $N^+(\infty)\cup N^+(0)$ contains at least
    $p+\frac{p-1}2+1 = \frac{3p+1}2$
    vertices and thus $Tr(QR_p)$ has Property $C_{2,\frac{3p+1}2}$.

  \item Let us now consider Property $C_{3,\frac{7p+3}4}$. We have
    to prove that given three vertices $x$, $y$, and $z$ of
    $Tr(QR_p)$ that form a triangle, we have
    $|N^+(x)\cup N^+(y) \cup N^+(z)|\ge \frac{7p+3}4$.

    We have to consider two cases depending on whether $x,y,z$ form a
    transitive triangle or $x,y,z$ form a directed triangle. By
    triangle-transitivity of $Tr(QR_p)$, it suffices to consider the
    cases $x,y,z = 0,1,\infty$ (transitive triangle) and
    $x,y,z=0,1',\infty$ (directed triangle).

    \begin{description}
    \item[Case $x,y,z = 0,1,\infty$: ] Let
      $A = N^+(1) \setminus \acc{N^+(0)\cup N^+(\infty)}$. We clearly
      have
      $|N^+(0)\cup N^+(1) \cup N^+(\infty)| = |N^+(0)\cup N^+(\infty)|
      + |A|$.  Since we already know that
      $|N^+(0)\cup N^+(\infty)|=\frac{3p+1}2$ (see the previous point), let us
      focus on the set $A$. We have
      $N^+(1) = N^+_{QR_p}(1) \uplus N^+_{QR'_p}(1) \uplus
      \acc{\infty}$.  Since $N_{QR'_p}^+(1) \subset N^+(\infty)$ and
      $\acc{\infty} \subset N^+(0)$, we have
      $A = N_{QR_p}^+(1) \setminus \acc{N^+(0) \cup
        N^+(\infty)}$. Since vertex $\infty$ has no out-neighbor in
      $QR_p$, we have $A = N_{QR_p}^+(1) \setminus N^+(0)$, that
      corresponds to the set of out-neighbors of $1$ in $QR_p$ which
      are not out-neighbors of $0$. Since $QR_p$ is a tournament, the
      vertices which are not out-neighbors of a given vertex are the
      in-neighbors of this vertex. Therefore, the set $A$ corresponds
      to the set of out-neighbors of $1$ in $QR_p$ which
      are in-neighbors of $0$ and thus $A =  N_{QR_p}^+(1) \cap
      N^-(0)$. This set has already been considered in the proof of
      \Cref{prop:qr11} where we showed that $|A| = \frac{p+1}4$. 

      Therefore, the set $N^+(0)\cup N^+(1)\cup N^+(\infty)$ contains
      $|N^+(0)\cup N^+(\infty)| + |A| \ge \frac{3p+1}2+\frac{p+1}4 =
      \frac{7p+3}4$ vertices.

    \item[Case $x,y,z = 0,1',\infty$: ] Note that $N^+(1') =
      N^-(1)$. Thus
      $N^+(0)\cup N^+(1')\cup N^+(\infty) = N^+(0)\cup N^-(1)\cup
      N^+(\infty)$. Using the same kind of arguments as previous case,
      we get that $|N^+(0)\cup N^-(1)\cup N^+(\infty)| \ge \frac{7p+3}4$.
    \end{description}
  \end{itemize}  
\end{proof}

\section{Graphs with maximum degree 3}\label{sec:deg3}

In this section, we consider graphs with maximum degree $3$ and we
prove that they all admit a homomorphism to the same target graph on
nine vertices.

Duffy et al.~\cite{duffy_oriented_2019} proved that every
\emph{connected} graph with maximum degree~3 has an oriented chromatic
number at most 9. To achieve this bound, they use the Paley tournament
$QR_{7}$ which has vertex set $V(QR_{7}) = \{0,1,\ldots,6\}$ and
$\vt{uv}\in A(QR_{7})$ whenever $v-u\equiv r\pmod{7}$ for
$r\in\{1,2,4\}$ (see~\Cref{subfig:qr7}).  Here is a quick sketch of
their proof. They first consider $2$-degenerated graphs with maximum
degree~3 (not necessarily connected) and prove the following:
\begin{theorem}~\cite{duffy_oriented_2019}\label{thm:chris}
  Every $2$-degenerate graph with maximum degree $3$ which does not contain
  a $3$-source adjacent to a $3$-sink is $QR_7$-colorable.
\end{theorem}

Then, given a \emph{connected} graph $G$ with maximum degree~3, they
firstly consider the case where $G$ contains a $3$-source. By removing
all $3$-sources from $G$, we obtain a graph $G'$ that is
$QR_7$-colorable by~\Cref{thm:chris}. It is then easy to put back all
the $3$-sources and color them with a new color $7$. Then,
subsequently, they consider the case where $G$ does not contain
$3$-sources. Removing any arc $\vt{uv}$ from $G$ leads to a graph $G'$
which admits a $QR_7$-coloring $\varphi$ by~\Cref{thm:chris}. To
extend $\varphi$ to $G$, it suffices to recolor $u$ and $v$ with two
new colors so that $\varphi(u)=7$ and $\varphi(v)=8$. This gives that
$G$ has an oriented chromatic number at most $9$.

The condition of connectivity of $G$ is a necessary condition in their
proof. Indeed, given a graph $G$ with maximum degree~3 which is not
connected, we need to remove one arc $\vt{u_iv_i}$ from each
$3$-regular component $C_i$ of $G$ to get a $2$-degenerate graph $G'$
which is $QR_7$-colorable by~\Cref{thm:chris}. However, to extend the
coloring to $G$ using two new colors, say color $7$ for the $u_i$'s
and color $8$ for the $v_i$'s, the colorings of each component must
agree on the neighbors of each $u_i$'s and on the neighbors of each
$v_i$'s, which is not necessarily the case. This potentially leads to
different target graphs on nine vertices for each
component. Therefore, even if each component has an oriented chromatic
number at most $9$, the whole graph may have an oriented chromatic
number strictly greater than $9$.

\bigskip

In the following, we prove this is not the case by showing that it is possible to color each component
with the same target graph $T_9$ on nine vertices whose construction
is described below. This implies that the condition of connectivity is
no more needed. 

\begin{figure}[h]
  \begin{center}
    \subfloat[$QR_7$.\label{subfig:qr7}]{\includegraphics[scale=0.7]{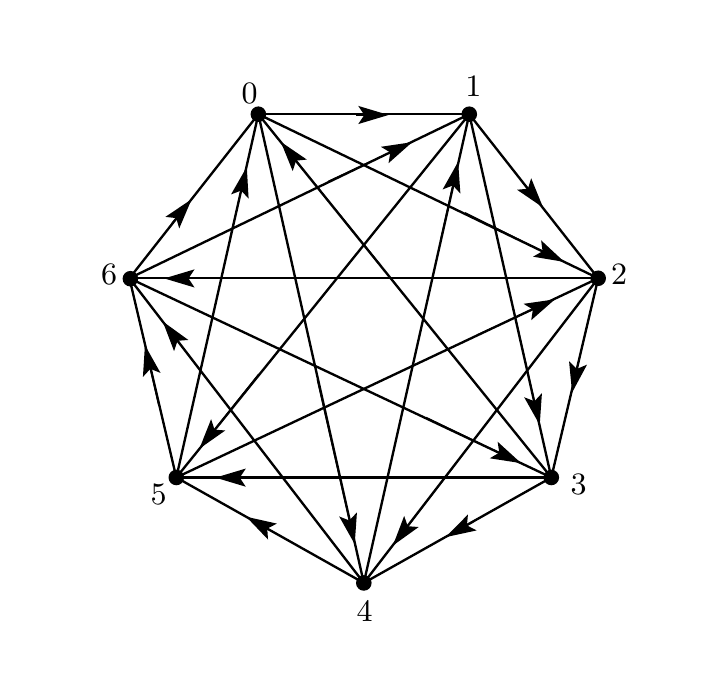}}
    \subfloat[$T_9$.\label{subfig:t9}]{\includegraphics[scale=0.7]{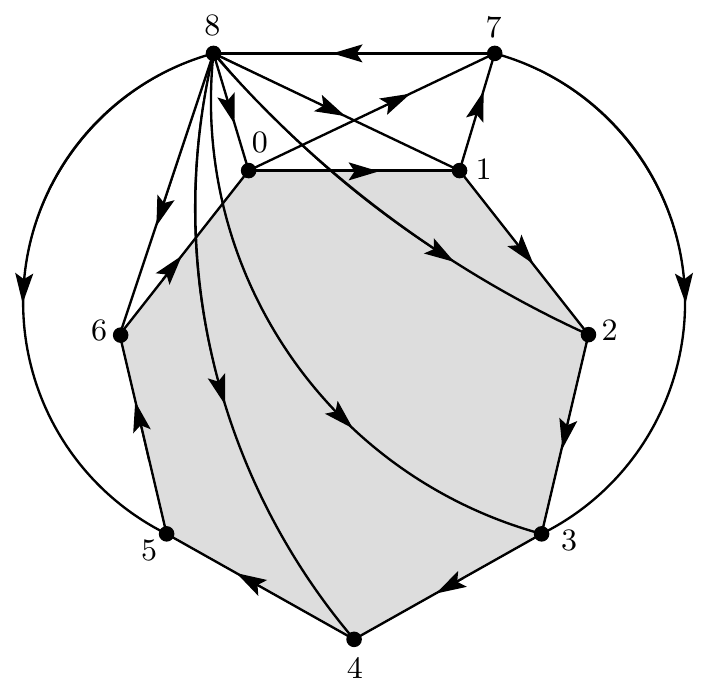}}
  \end{center}
  \caption{The oriented graphs $QR_7$ and $T_9$.}
  \label{fig:qr7t9}
\end{figure}

The oriented graph $T_9$ is obtained from $QR_7$
(see~\Cref{subfig:qr7}) by adding two vertices labelled~$7$ and~$8$,
and the arcs $\vt{07}$, $\vt{17}$, $\vt{73}$, $\vt{78}$, and $\vt{8i}$
for every $0\le i\le6$ (see~\Cref{subfig:t9} where the grey part
stands for $QR_7$).

We prove the following:
\begin{theorem}\label{thm:deg3}
  Every graph with maximum degree $3$ admits a $T_9$-coloring.
\end{theorem}

\begin{proof}
  It is sufficient to show that every connected graph $G$
  with maximum degree $3$ admits a $T_9$-coloring. 
  We consider the following cases.
  \begin{itemize}
  \item We suppose that $G$ is $2$-degenerate or $G$ contains a
    $3$-source. Let $G'$ be the oriented graph obtained from $G$ by
    removing every $3$-source. Since $G'$ is $2$-degenerate and
    contains no $3$-source, $G'$ admits a $QR_7$-coloring $\varphi$
    by~\Cref{thm:chris}. We extend $\varphi$ to a $T_9$-coloring of
    $G$ by setting $\varphi(u)=8$ for every $3$-source $u$ of $G$
    (indeed, the vertex $8$ of $T_9$ dominates all the vertices of
    $QR_7$).
    \begin{figure}[h]
      \centering
      \includegraphics{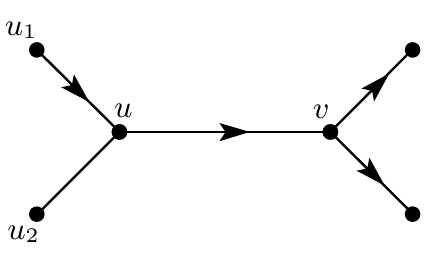}
      \caption{Configuration of~\Cref{thm:deg3}.}
      \label{fig:config-deg3}
    \end{figure}

  \item We suppose that $G$ is $3$-regular and contains no $3$-source.
    Notice that $G$ necessarily contains a vertex $v$ of out-degree
    two. Let $u$ denote the in-neighbor of $v$. Since $u$ is not a
    $3$-source, it has an in-neighbor $u_1$. Let $u_2$ denote the
    neighbor of $u$ distinct from $u_1$ and $v$
    (see~\Cref{fig:config-deg3}). We consider the graph $G'$ obtained
    from $G$ by removing the arc $\vt{uv}$. Since $G'$ is
    $2$-degenerate and contains no $3$-source, $G'$ admits a
    $QR_7$-coloring $\varphi$ by~\Cref{thm:chris}.
    \begin{itemize}
    \item If $G$ (or equivalently $G'$) contains the arc $\vt{uu_2}$,
      then we necessarily have $\varphi(u_1)\ne \varphi(u_2)$. If
      $\vt{\varphi(u_1)\varphi(u_2)}\in A(QR_7)$
      (resp. $\vt{\varphi(u_2)\varphi(u_1)}\in A(QR_7)$), we recolor
      $G'$ so that $\varphi(u_1)=1$ (resp. $\varphi(u_1)=0$) and
      $\varphi(u_2)=3$ by the arc-transitivity of $QR_7$. It can be
      easily checked that we can extend $\varphi$ to a $T_9$-coloring
      of $G$ by setting $\varphi(u)=7$ and $\varphi(v)=8$.

    \item If $G$ contains $\vt{u_2u}$, then by the arc-transitivity of
      $QR_7$, we can assume that $\acc{\varphi(u_1),\varphi(u_2)}$
      $\subseteq\acc{0,1}$. Again, $\varphi$ can be extended to
      $T_9$-coloring of $G$ by setting $\varphi(u)=7$ and
      $\varphi(v)=8$.
    \end{itemize}
  \end{itemize}
\end{proof}

\section{Graphs with maximum degree at least $4$}\label{sec:deg4}

In this section, we consider graphs with maximum degree at least~$4$.

Duffy et al.~\cite{duffy_oriented_2019} recently proved that
every \emph{connected} graph with maximum degree~$4$ has an oriented
chromatic number at most $69$. To prove their result, they consider
the case of $3$-degenerate graphs with maximum degree~$4$ and prove
that they admit a homomorphism to the Paley tournament $QR_{67}$ on $67$ vertices; they
then show how to extend such a $67$-coloring to \emph{connected} graphs with maximum
degree~$4$ using two more colors, leading to a $69$-coloring.

We propose a general result which determines properties of a target
graph to be universal for graphs of maximum degree $\Delta\ge 4$. As
for graphs with maximum degree~$3$ (see~\Cref{sec:deg3}), the
condition of connectivity is not needed. As an example, our general
result substantially decreases the bound of~$69$ colors for graphs
with maximum $4$ due to Duffy et al.~\cite{duffy_oriented_2019}
to~$26$ colors.

\begin{theorem}\label{thm:deg5-8-degenerate}
  Every $(\Delta-1)$-degenerate graph with maximum degree $\Delta\ge 2$
  admits a $T$-coloring where $T$ is a graph on $n$ vertices with
  Properties $P_{\Delta-1,\Delta-2}$ and
  $C_{\Delta-2,\frac{n(\Delta-2)}{\Delta-1}+1}$.
\end{theorem}

\begin{proof}
  Let $G$ be a minimal counter-example to
  \Cref{thm:deg5-8-degenerate}. By definition, $G$ contains a $k$-vertex $u$ with
  $k\le \Delta-1$. Let $v_1,v_2,\dots,v_k$ be the neighbors of $u$.

  Suppose first that there exists one arc in the neighborhood of $u$,
  and w.l.o.g. assume that $\vt{v_1v_2}$ is an arc of $G$. The graph
  $G'$ obtained from $G$ by removing the arc $uv_1$ admits a
  $T$-coloring~$\varphi$ by minimality of $G$. Vertex $v_1$ has at
  most $\Delta-1$ colored neighbors $w_1,\dots,w_{\Delta-1}$ by
  $\varphi$ ($u$ is uncolored). Note that some $w_i$'s may coincide
  with some $v_i$'s. The sequence
  $S=(\varphi(w_1),\varphi(w_2),\dots,\varphi(w_{\Delta-1}))$ is
  compatible w.r.t. the orientations of the arcs $w_iv_1$ since $v_1$
  and every $w_i$ are colored in $G'$. We uncolor $v_1$ and by
  Property $P_{\Delta-1,\Delta-2}$ of $T$, we have $\Delta-2$
  available colors for $v_1$ among which at least one, say color $c$,
  is compatible with $\varphi(v_3),\varphi(v_4),\dots,\varphi(v_{k})$
  w.r.t. to the orientation of the arcs $uv_1, uv_3, uv_4,\dots,uv_k$.
  Set $\varphi(v_1)=c$. Note that we necessarily have
  $c\notin\acc{\varphi(v_2),\at(\varphi(v_2))}$ since $\vt{v_1v_2}$ is
  an arc of $G$ and thus $c$ is compatible with
  $\varphi(v_2)$. Note also that
  the sequence $S=(\varphi(v_2),\varphi(v_3),\dots,\varphi(v_{k}))$ is
  compatible w.r.t. the orientations of the arcs $uv_i$,
  $2\le i\le k$, since vertices $v_2,v_3,\dots,v_k$ are colored in
  $G'$. Therefore, the sequence
  $S=(\varphi(v_1),\varphi(v_2),\dots,\varphi(v_{k}))$ is compatible
  with the orientation of the arcs $uv_i$, $1\le i\le k$. By Property
  $P_{\Delta-1,\Delta-2}$, we have at least $\Delta-2$ available
  colors for $u$. Therefore, $\varphi$ can be extended to a
  $T$-coloring of $G$, a contradiction.
  
  Assume now that there is no arc in the neighborhood of $u$. The
  graph $G'$ obtained from $G$ by removing the vertex $u$ 
  admits a $T$-coloring $\varphi$ by minimality of $G$. Let
  $w_i$, $1\le i \le \Delta-1$, be the neighbors of $v_1$ distinct from
  $u$. Since $v_1$ and every $w_i$ are colored in $G'$, the sequence
  $S=(\varphi(w_1),\varphi(w_2),\dots,\varphi(w_{\Delta-1}))$ is compatible
  with the orientations of the arcs $vw_i$ and by Property
  $P_{\Delta-1,\Delta-2}$, we have $\Delta-2$ available colors
  for $v_1$. Note that these $\Delta-2$ colors
  (which are vertices of $T$) form
  a clique subgraph of $T$. The same holds
  for every $v_i$ (i.e. each $v_i$ has  $\Delta-2$ available colors)
  and since there is no arc in the neighborhood of 
  $u$, each $v_i$ can be colored independently.

  By~Property $C_{\Delta-2,\frac{n(\Delta-2)}{\Delta-1}+1}$, given
  these $\Delta-2$ possible colors for $v_1$, there are at least
  $\frac{n(\Delta-2)}{\Delta-1}+1$ choices of colors for $u$; thus
  $v_1$ forbids at most $n-\paren{\frac{n(\Delta-2)}{\Delta-1}+1}=\frac{n}{\Delta-1}-1$
  colors for $u$. By the same argument, each $v_i$ forbids
  at most $\frac{n}{\Delta-1}-1$ colors for $u$. Therefore, $u$ has at
  most
  $k(\frac{n}{\Delta-1}-1) \le (\Delta-1)(\frac{n}{\Delta-1}-1) =
  n-\Delta+1 < n$ forbidden colors since $\Delta\ge 2$.  This means
  that there are at least one available color for $u$. Therefore,
  $\varphi$ can be extended to a $T$-coloring of $G$, a contradiction.
\end{proof}

To achieve our bounds on chromatic oriented number, we apply Tromp's construction to Paley
tournaments $QR_{p}$.
Moreover, given a graph $Tr(QR_p)$, we construct the graph $Tr^*(QR_p)$ on $2p+4$
vertices by adding
two vertices $t_0$ and $t_1$ such that $t_0$ is a twin vertex of
vertex $0$ (\ie a vertex with the same neighborhood as vertex $0$) and
$t_1$ is a twin vertex of vertex $1$. We finally add the arc
$\vt{t_1t_0}$. 

\begin{theorem}\label{thm:deg5-8}
  Every graph
  with maximum degree $\Delta\ge 4$ admits a $Tr^*(QR_p)$-coloring where
  $Tr(QR_p)$ is a Tromp graph (built from a Paley tournament $QR_p$)
  with Properties $P_{\Delta-1,\Delta-2}$ and \linebreak
  $C_{\Delta-2,\frac{(2p+2)(\Delta-2)}{\Delta-1}+1}$.
\end{theorem}

\begin{proof}
  Since $Tr(QR_p)$ is a subgraph of $Tr^*(QR_p)$, it remains to prove that
  every connected $\Delta$-regular graph $H$ admits a $Tr^*(QR_p)$-coloring.
  
  Let $H'=H\setminus\acc{\vt{uv}}$ where $\vt{uv}$ is any arc of
  $H$. By~\Cref{thm:deg5-8-degenerate}, $H'$ admits a
  $Tr(QR_p)$-coloring~$\varphi$. By vertex-transitivity of $Tr(QR_p)$, we may
  assume that $\varphi(v)=0$. By Property~$P_{\Delta-1,\Delta-2}$ of $Tr(QR_p)$, we have $\Delta-2$
  available colors $c_1,c_2,\dots,c_{\Delta-2}$ for $u$.
  Note that, given $1\le i < j \le \Delta-2$, we necessarily
  have $c_i\neq \at(c_j)$. We thus recolor $u$ with one of the $c_i$'s
  so that $\varphi(u) \notin \acc{0,0'}$ since $\Delta\ge 4$.
  Therefore, $\varphi(u)\varphi(v)$ is an arc of $Tr(QR_p)$.
  
  If $\vt{\varphi(u)\varphi(v)}$ is an arc of $Tr(QR_p)$, then $\varphi$
  is a $Tr(QR_p)$-coloring of $H$ and thus a $Tr^*(QR_p)$-coloring of
  $H$. Therefore, $\vt{\varphi(v)\varphi(u)}$ is an arc of $Tr(QR_p)$. By
  arc-transitivity of $Tr(QR_p)$, we may assume that $\varphi(u)=1$. We
  recolor $u$ and $v$ so that $\varphi(u)=t_1$ and
  $\varphi(v)=t_0$. Since each $t_i$ is a twin vertex of vertex $i$ in
  $Tr^*(QR_p)$ and $\vt{t_1t_0}$ is an arc $Tr^*(QR_p)$, it is easy to verify
  that $\varphi$ is now a $Tr^*(QR_p)$-coloring of $H$.
\end{proof}

The Properties \pnk of Paley tournaments $QR_p$ can be expressed
by a formula for $n\le2$ (see \Cref{prop:pnk}). For higher values of
$n$, there are no formula and some
properties are only known for small values of $p$. In our case, we are
interested in properties of the form $P_{n,n}$. We
computed these properties using a computer-check.

\begin{proposition}\label{thm:properties}
  The smallest Paley tournament with Property $P_{2,2}$
  (resp. $P_{3,3}$, $P_{4,4}$, $P_{5,5}$) is $QR_{11}$
  (resp. $QR_{43}$, $QR_{151}$, $QR_{659}$).
\end{proposition}

\begin{proof}
  By computer.
\end{proof}

Determining more of such properties would be possible with more
computing power and time. Note that Paley tournament $QR_{151}$ has in
fact Property $P_{5,6}$ and there exist no smaller Paley tournament
with Property $P_{5,5}$. 

As a corollary of \Cref{prop:pnk,thm:properties}, we get the following.
\begin{corollary}\label{coro:tromp}
  The Tromp graph $Tr(QR_{11})$ (resp. $Tr(QR_{43})$, $Tr(QR_{151})$,
  $Tr(QR_{659})$) has Property $P_{3,2}$ (resp. $P_{4,3}$, $P_{5,4}$,
  $P_{6,5}$).
\end{corollary}

As corollaries of \Cref{rem:cnk,lem:17,thm:deg5-8,coro:tromp}, we
obtain the following four results. We give a proof of the last one (\Cref{coro:proof}),
the other three corollaries follow the same arguments.

\begin{corollary}
  Every graph $G$ with maximum degree $4$ admits a
  $Tr^*(QR_{11})$-coloring. Thus, $\chi_o(G) \le 26$.
\end{corollary}

\begin{corollary}
  Every graph $G$ with maximum degree $5$ admits a
  $Tr^*(QR_{43})$-coloring. Thus, $\chi_o(G) \le 90$.
\end{corollary}

\begin{corollary}
  Every graph $G$ with maximum degree $6$ admits a
  $Tr^*(QR_{151})$-coloring. Thus, $\chi_o(G) \le 306$.
\end{corollary}

\begin{corollary}\label{coro:proof}
  Every graph $G$ with maximum degree $7$ admits a
  $Tr^*(QR_{659})$-coloring. Thus, $\chi_o(G) \le 1322$.
\end{corollary}

\begin{proof}
  Let $\Delta=7$. By \Cref{lem:17}, the graph $Tr(QR_{659})$ has
  Property $C_{3,\frac{7p+3}4} = C_{3,1154}$. By \Cref{rem:cnk}, it thus
  has Property $C_{5,1101} = C_{\Delta-2,\frac{(2p+2)(\Delta-2)}{\Delta-1}+1} 
  $. By \Cref{coro:tromp}, it also has Property
  $P_{6,5} = P_{\Delta-1,\Delta-2}$. The graph $Tr(QR_{659})$ verifies
  the hypothesis of \Cref{thm:deg5-8}, and thus every graph with
  maximum degree $7$ admits a $Tr^*(QR_{659})$-coloring.
\end{proof}

\end{document}